\documentclass{llncs}
\usepackage{ifthen}
\usepackage{amsmath}
\usepackage{amssymb}
\usepackage{amsfonts}
\usepackage{amstext}
\usepackage{xspace}
\usepackage[noend]{algorithmic}
\usepackage[boxed]{algorithm}

\usepackage{color}
\usepackage{enumitem}
\usepackage{hyperref}

\usepackage{microtype}

\newcommand{\subsetsum}{{\text{\sc Subset Sum}}\xspace}

\newcommand{\linearsat}{{\text{\sc Linear Sat}}\xspace}

\newcommand{\supp}{\mathrm{supp}}

\newcommand{\polylog}{\operatorname{polylog}}
\newcommand{\rank}{\operatorname{rk}}

\newcommand{\no}{\textbf{no}}
\newcommand{\yes}{\textbf{yes}}

\newcommand{\omitted}{$\dagger$}

\newcommand{\suppbound}{S}
\newcommand{\OH}{\mathcal{O}}
\newcommand{\OHT}{\mathcal{O}}
\newcommand{\OHS}{\mathcal{O}^\star}
\newcommand{\modulus}{p}

\newcommand{\prob}{\operatorname{Prob}}

\newcommand{\bbn}{\mathbb{N}}
\newcommand{\bbz}{\mathbb{Z}}

\DeclareMathSymbol{\Gamma}{\mathalpha}{operators}{0}
\DeclareMathSymbol{\Delta}{\mathalpha}{operators}{1}
\DeclareMathSymbol{\Theta}{\mathalpha}{operators}{2}
\DeclareMathSymbol{\Lambda}{\mathalpha}{operators}{3}
\DeclareMathSymbol{\Xi}{\mathalpha}{operators}{4}
\DeclareMathSymbol{\Pi}{\mathalpha}{operators}{5}
\DeclareMathSymbol{\Sigma}{\mathalpha}{operators}{6}
\DeclareMathSymbol{\Upsilon}{\mathalpha}{operators}{7}
\DeclareMathSymbol{\Phi}{\mathalpha}{operators}{8}
\DeclareMathSymbol{\Psi}{\mathalpha}{operators}{9}
\DeclareMathSymbol{\Omega}{\mathalpha}{operators}{10}

\let\doendproof\endproof
\renewcommand\endproof{~\hfill\qed\doendproof}

\newcommand{\fieldF}{\bbf}

\newcommand{\st}[2]{\left\langle #1,#2 \right\rangle} 

\def\shortV{false}
\newcommand{\sv}[2]{\ifthenelse{\equal{\shortV}{true}}{#1}{#2}}

\newtheorem{thm}{Theorem}
\newtheorem{lem}[thm]{Lemma}
\newtheorem{obs}[thm]{Observation}
\newtheorem{cor}[thm]{Corollary}
\newtheorem{defi}[thm]{Definition}
\newtheorem{clm}{Claim}

\renewcommand{\mod}{\operatorname{mod}}
\renewcommand{\vec}[1]{\boldsymbol #1}
\newcommand{\mat}[1]{\boldsymbol #1}

\title{Homomorphic Hashing for\\ Sparse Coefficient Extraction}
\author{ Petteri Kaski\inst{1} \and Mikko Koivisto\inst{2} \and Jesper Nederlof\inst{3}}

\institute{Helsinki Institute for Information Technology, Department of Computer Science, Aalto University, Finland. \href{mailto:petteri.kaski@aalto.fi}{petteri.kaski@aalto.fi}. Supported by the Academy of Finland, Grants 252083 and 256287. \and Helsinki Institute for Information Technology, Department of Computer Science University of Helsinki, Finland. \href{mailto:mkhkoivi@cs.helsinki.fi}{mkhkoivi@cs.helsinki.fi} Supported by the Academy of Finland, Grant 125637. \and Utrecht University, Utrecht, The Netherlands. \href{mailto:j.nederlof@uu.nl}{j.nederlof@uu.nl}. Supported by the Nederlandse Organisatie voor Wetenschappelijk Onderzoek (NWO), project: 'Space and Time Efficient Structural Improvements of Dynamic Programming Algorithms'.}

\begin{document}
\maketitle
\addtocounter{footnote}{-3}
\vspace{-2em}
\begin{abstract}

We study classes of Dynamic Programming (DP) algorithms which, due to their algebraic definitions, are closely related to coefficient extraction methods. DP algorithms can easily be modified to exploit sparseness in the DP table through memorization. Coefficient extraction techniques on the other hand are both space-efficient and parallelisable, but no tools have been available to exploit sparseness. 
We investigate the systematic use of homomorphic hash functions to combine the best of these methods and obtain improved space-efficient algorithms for problems including LINEAR SAT, SET PARTITION, and SUBSET SUM. Our algorithms run in time proportional to the number of nonzero entries of the last segment of the DP table, which presents a strict improvement over sparse DP. The last property also gives an improved algorithm for CNF SAT with sparse projections.
\end{abstract}

\section{Introduction}

{\em Coefficient extraction} can be seen as a general method for designing algorithms, recently in particular in the area of exact algorithms for various NP-hard problems \cite{DBLP:conf/focs/Bjorklund10,DBLP:journals/siamcomp/BjorklundHK09,DBLP:conf/icalp/KoutisW09,DBLP:conf/stoc/LokshtanovN10,DBLP:conf/icalp/Nederlof09,DBLP:journals/ipl/Williams09} (cf. \cite{Fominbook,DBLP:journals/dam/Woeginger08} for an introduction to exact algorithms). 
The approach of the method is the following (see also~\cite{lipton10}): 
\begin{enumerate}
	\item Define a variable (the so-called coefficient) whose value (almost) immediately gives the solution of the problem to be solved, 
	\item Show that the variable can be expressed by a relatively small formula or circuit over a (cleverly chosen) large algebraic object like a ring or field,
	\item Show how to perform operations in the algebraic object relatively efficiently.
\end{enumerate}
In a typical application of the method, the first two steps are derived from an existing Dynamic Programming (DP) algorithm, and the third step deploys a carefully selected algebraic isomorphism, such as the discrete Fourier transform to extract the desired solution/coefficient. Algorithms based on coefficient extraction have two key advantages over DP algorithms; namely, they are space-efficient and they parallelise well (see, for example, \cite{DBLP:conf/stoc/LokshtanovN10}).

Yet, DP has an advantage if the problem instance is {\em sparse}. By this we mean that the number of candidate/partial solutions that need to be considered during DP is small, that is, most entries in the DP table are not used at all. In such a case we can readily adjust the DP algorithm to take this into account through {\em memorization} so that both the running time and space usage become proportional to the number of partial solutions considered. Unfortunately, it is difficult to parallelise or lower the space usage of memorization. Coefficient extraction algorithms relying on interpolation of sparse polynomials \cite{DBLP:journals/siamcomp/Mansour95} improve over memorization by scaling proportionally only to the number of {\em candidate} solutions, but their space usage is still not satisfactory (see also \cite{DBLP:journals/dam/Woeginger08}).

This paper aims at obtaining what is essentially the best of both worlds, by investigating the systematic use of homomorphisms to ``hash down'' circuit-based coefficient extraction algorithms so that the domain of coefficient extraction -- and hence the running time -- matches or improves that of memorization-based DP algorithms, while providing space-efficiency and efficient parallelisation. 
The key idea is to take an existing algebraic circuit for coefficient extraction (over a sparsely populated algebraic domain such as a ring or field), and transform the circuit into a circuit over a smaller domain by a homomorphic hash function, and only then perform the actual coefficient extraction. Because the function is homomorphic, by hashing the values at the input gates and evaluating the circuit, the output evaluates to the hash of the original output value. Because the function is a hash function, the coefficient to be extracted collides with other coefficients only with negligible probability in the smaller domain, and coefficient extraction can be successfully used on the new (hashed-down) circuit. We call this approach {\em homomorphic hashing}. 

\subsection*{Our and previous results}

We study sparse DP/coefficient extraction in three domains: (a) the univariate polynomial ring $\fieldF[x]$ in Section~\ref{sec:sss}, (b) the group algebra $\fieldF[\bbz_2^n]$ in Section~\ref{sec:setpart} and (c) the M\"obius algebra of the subset lattice in Section~\ref{sec:unionhashing}. The subject of sparse DP or coefficient extraction is highly motivated and well-studied \cite{Eppstein:1992:SDP:146637.146650,EppGalGia-JACM-92-II,DBLP:journals/siamcomp/Mansour95,zippel}. In~\cite{DBLP:journals/siamcomp/Mansour95}, a sparse polynomial interpolation algorithm using exponential space was already given for (a) and (b); our algorithms improve these domains to polynomial space. In~\cite{DBLP:conf/icalp/KoutisW09} a polynomial-space algorithm for finding a small multilinear monomial in $\fieldF[\bbz_2^n]$ was given. In~\cite{DBLP:conf/stoc/LokshtanovN10} a general study of settings (a) and (b) was initiated, but sparsity was not addressed. Our main technical contribution occurs with (c) and hashing down to the ``Solomon algebra'' of a poset.

Our methods work for general arithmetic circuits similarly as in \cite{DBLP:conf/icalp/KoutisW09,DBLP:conf/stoc/LokshtanovN10,DBLP:journals/siamcomp/Mansour95}, and most of our algorithms work for counting variants as well. But, for concreteness, we will work here with specific decision problems. Although we mainly give improvements for sparse variants of these problems, we feel the results will be useful to deal with the general case as well (as we will see in Section~\ref{sec:setpart}).

\paragraph{Subset Sum}

The \subsetsum problem is the following: given a vector $\vec{a}=(a_1,\ldots,a_n)$ and integer $t$, determine whether there exists a subset $X \subseteq [n]$ such that $\sum_{e \in X}a_e=t$. It is known to be solvable $\OHS(2^{n/2})$ time and $\OHS(2^{n/4})$ space~\cite{Horowitz74,DBLP:journals/siamcomp/SchroeppelS81}, and solving it faster, or even in $\OHS(1.99^n)$ time and polynomial space are interesting open questions~\cite{DBLP:journals/dam/Woeginger08}. Recently, a polynomial space algorithm using $\OHS(t)$ time was given in~\cite{DBLP:conf/stoc/LokshtanovN10}.

\begin{thm}\label{thm:mainsss}
An instance $(\vec{a},t)$ of the \subsetsum problem can be solved
	\begin{enumerate}[label=(\alph*)]
		\item in $\OHS(S)$ expected time and polynomial space, and
		\item in $\OHS(S^2)$ time and polynomial space,
	\end{enumerate}
where $S$ is the number of distinct sums, i.e. $S=|\{ \sum_{e \in X}a_e : X \subseteq [n]\}|$ .
\end{thm}

Here it should be noted that standard sparse DP gives an $\OHS(S)$ time and space algorithm. Informally stated, our algorithms hash the instances by working modulo randomly chosen prime numbers and applying the algoritm of~\cite{DBLP:conf/stoc/LokshtanovN10}. While interesting on their own, these results may be useful in resolving the above open questions when combined with other techniques.

\paragraph{Linear Sat}
The {\sc Linear Sat} problem is defined as follows: given a matrix $\mat{A} \in \bbz_2^{n \times m}$, vectors $\vec{b}\in \bbz_2^m$ and $\vec{\omega} \in \bbn^n$, and an integer $t = n^{\OHT(1)}$, determine whether there is there a vector $\vec{x} \in \bbz_2^n$ such that $\vec{x}\mat{A}=\vec{b}$ and $\vec{\omega}\vec{x}^T \leq t$.

Variants of \textsc{Linear Sat} have been studied, perhaps most notably in \cite{Hastad01}, where approximability was studied. In \cite{Alon10solvingmax-r-sat,DBLP:journals/corr/abs-1104-1135} the Fixed Parameter Tractability was studied for parameterizations of various above guarantees. There, it was also quoted from \cite{Hastad01} that (a variant of) \textsc{Linear Sat} is ``as basic as satisfiability''.



It can be observed that using the approach from~\cite{Horowitz74}, \textsc{Linear Sat} can be solved in $\OHT(2^{n/2}m)$ time and $\OHT(2^{n/2}m)$ space. Also, using standard ``sparse dynamic programming'', it can be solved in $\OHS(2^{\rank(\mat{A})})$ time and $\OHS(2^{\rank(\mat{A})})$ space, where $\rank(\mat{A})$ is the rank of $\mat{A}$. We obtain the following polynomial-space variants:



\begin{thm}\label{thm:mainls}
An instance $(\mat{A},\vec{b},\vec{\omega},t)$ of \linearsat can be solved by algorithms with constant one-sided error probability in 
	\begin{enumerate}[label=(\alph*)]
		\item $\OHS(2^{\rank(\mat{A})})$ time and polynomial space, and
		\item $\OHS(2^{n /2})$ time and polynomial space.
	\end{enumerate}
\end{thm}

The first algorithm hashes the input down using a random linear map and afterwards determines the answer using the Walsh-Hadamard transform. The second algorithm
uses a Win/Win approach, combining the first algorithm with the fact that an $\mat{A}$ with high rank can be solved with a complementary algorithm.

\paragraph{Satisfiability}
The {\sc CNF-Sat} problem is defined as follows: given a CNF-formula $\phi = C_1 \wedge C_2 \wedge \ldots \wedge C_m$ over $n$ variables, determine whether $\phi$ is satisfiable. There are many interesting open questions related to this problem, a major one being whether it can be solved in time $\OHS((2-\epsilon)^n)$ (the `Strong Exponential Time Hypothesis' \cite{DBLP:journals/jcss/ImpagliazzoP01} states this is not possible), and another being whether the number of satisfying assignments can be counted in time $\OHS((2-\epsilon)^n)$ for some $\epsilon > 0$ (e.g. \cite{Traxler10}).

A \emph{prefix assignment} is an assignment of 0/1 values to the variables $v_1,\ldots,v_i$ for some $1\leq i\leq n$. A \emph{projection} (\emph{prefix projection}) of a CNF-formula is a subset $\pi \subseteq [m]$ such that there exists an assignment (prefix assignment) of the variables such that for every $1 \leq j \leq m$ it satisfies $C_j$ if and only if $j \in \pi$. An algorithm for {\sc CNF-Sat} running in time linear in the number of \emph{prefix} projections can be obtained by standard sparse DP. However, it is sensible to ask about complexity of {\sc CNF-Sat} if the number of projections is small. We give a positive answer:

\begin{thm}\label{thm:maincnf}
Satisfiability of a formula $\phi=C_1\wedge \ldots \wedge C_m$ can be determined in $\OHS(P^2)$ time and $\OHS(P)$ space, where $P=|\{\pi \subseteq [m]: \pi \text{ is a projection of } \phi \}|$.
\end{thm}

We are not aware of previous work that studies instances with few
projections. Although most instances will have many projections, we think our
result opens up a fresh technical perspective that may contribute towards 
solving the above mentioned and related questions.

Underlying Theorem~\ref{thm:maincnf} is our main technical contribution (Theorem~\ref{thm:unionhashsmallsupp}) that enables us to circumvent partial projections and access projections directly, namely homomorphic hashing from the M\"obius algebra of the lattice of subsets of $[m]$ to the Solomon algebra of a poset. A full proof of Theorem~\ref{thm:unionhashsmallsupp} is given in the appendix; we give a specialized, more direct proof of Theorem~\ref{thm:maincnf} in Section~5. 


The proofs of claims marked with a ``$\dagger$'' are relegated to the appendix
in order to not break the flow of the paper.

\section{Notation and Preliminaries}
\label{sec:not}
Lower-case boldface characters refer to vectors, while capital boldface letters refer to matrices, $\mat{I}$ being the identity matrix. The rank of a matrix $\mat{A}$ is denoted by $\rank(\mat{A})$. If $R$ and $S$ are sets, and $S$ is finite, denote by $R^S$ the set of all $|S|$-dimensional vectors with values in $R$, indexed by elements of $S$, that is, if $\vec{v}\in R^S$, then for every $e \in S$ we have $v_e \in R$. We denote by $\bbz$ and $\bbn$ the set of integers and non-negative integers, respectively, and by $\bbz_p$ the field of integers modulo a prime $p$. An arbitrary field is denoted by $\fieldF$. 

For a logical proposition $P$, we use Iverson's bracket notation $[P]$ to denote a $1$ if $P$ is true and a $0$ if $P$ is false. For a function $h:A \rightarrow B$ and $b \in B$, the preimage $h^{-1}(b)$ is defined as the set $\{a \in A: h(a)=b\}$. For an integer $n$ and $A \subseteq \{1,\ldots,n\}$, denote by $\vec\chi(A) \in \bbz_2^{n}$ the characteristic vector of $A$. Sometimes we will state running times of algorithms with the $\OHS$ notation, which suppresses any factor polynomial in the input size. 

For a ring $R$ and a finite set $S$, we write $R^S$ for the ring consisting of the set $R^S$ (the set of all vectors over $R$ with coordinates indexed by elements of $S$) equipped with coordinate-wise addition $+$ and multiplication $\circ$ (the \emph{Hadamard product}), that is, for $\vec{a},\vec{b} \in R^S$ and $\vec{a} + \vec{b} = \vec{c}$, $\vec{a} \circ \vec{b}=\vec{d}$ we set $a_z+b_z=c_z$ and $a_zb_z=d_z$ for each $z \in S$, where $+$ and the juxtaposition denote addition and multiplication in $R$, respectively. The inner-product $\vec{a},\vec{b} \in R^S$ is denoted by $\vec{a}^T\cdot\vec{b}$. For $\vec{v} \in R^S$ denote by $\supp(\vec{v}) \subseteq S$ the \emph{support} of $\vec{v}$, that is, $\supp(\vec{v})=\{z \in S:v_z \neq 0\}$, where $0$ is the additive identity element of $R$. A vector $\vec{v}$ is called a \emph{singleton} if $|\supp(\vec{v})|=1$. We denote by $\st{w}{z}$ the singleton with value $w$ on index $z$, that is, $\st{w}{z}_y=w[y=z]$ for all $y\in S$.

If $R$ is a ring and $(S,\cdot)$ is a finite semigroup, denote by $R[S]$ the ring consisting of the set $R^S$ equipped with coordinate-wise addition and multiplication defined by the convolution operator $*$, where for $\vec{a},\vec{b} \in R^S$, $\vec{a}*\vec{b}=\vec{c}$ we set $c_z = \sum_{x\cdot y=z}a_xb_y$ for every $z \in S$.

If $R,S$ are rings with operations $(+,*)$ and $(\oplus,\circledast)$ respectively, a \emph{homomorphism from $R$ to $S$} is a function $h:R \rightarrow S$ such that $h(e_1 + e_2)=h(e_1)\oplus h(e_2)$ and $h(e_1 * e_2)=h(e_1) \circledast h(e_2)$ for every $e_1,e_2 \in R$.

\begin{obs}
\label{obs:hom}
Let $R$ be a ring, and let $(S,\cdot)$ and $(T,\odot)$ be finite semigroups. Suppose $\varphi: S \rightarrow T$ such that for every $x,y \in S$ we have $\varphi(x\cdot y)=\varphi(x)\odot \varphi(y)$. Then the function $h: R[S] \rightarrow R[T]$ defined by $h:\vec{a} \mapsto \vec{b}$ where $b_z = \sum_{y\in \varphi^{-1}(z)}a_{y}$ for all $z\in T$ is a homomorphism.
\end{obs}

A \emph{circuit}
$C$ over a ring $R$ is a labeled directed acyclic graph $D=(V,A)$ where the elements of $V$ are called \emph{gates} and $D$ has a unique sink called the \emph{output gate of $C$}. All sources of $C$ are called \emph{input gates} and are labeled with elements from $R$. All gates with non-zero in-degree are labeled as either an \emph{addition} or a \emph{multiplication gate}. (If multiplication in $R$ is not commutative, the in-arcs of each multiplication gate are also ordered.) Every gate $g$ of $C$ can be associated with a ring element in the following natural way: If $g$ is an input gate, we associate the label of $g$ with $g$. If $g$ is an addition gate we associate the ring element $e_1+\ldots+e_d$ with $g$, and if $g$ is a multiplication gate we associate the ring element $e_1*\ldots* e_d$ with $g$ where $e_1,\ldots,e_d$ are the ring elements associated with the $d$ in-neighbors of $g$, and $+$ and $*$ are the operations of the ring $R$. 

Suppose the ground set of $R$ is of the type $A^B$ where $A,B$ are sets. Then $C$ is said to have \emph{singleton inputs} if the label of every input-gate of $C$ is a singleton vector of $R$.

\begin{defi}
Let $R$ and $S$ be rings, let $h:R\rightarrow S$ be a homomorphism, and suppose that $C$ is a circuit over $R$. Then, the circuit $h(C)$ over $S$ \emph{obtained by applying $h$ to $C$} is defined as the circuit obtained from $C$ by replacing for every input gate the label $l$ by $h(l)$.
\end{defi}

Note that the following is immediate from the definition of a homomorphism:

\begin{obs}
\label{obs:homcir}
Suppose $C$ is a circuit over a ring $R$ with output $v \in R$. Then the circuit over $S$ obtained by applying a homomorphism $h:R\rightarrow S$ to $C$ outputs $h(v)\in S$.
\end{obs}

\section{Homomorphic Hashing for Subset Sum}
\label{sec:sss}


In this section we will study the \subsetsum problem and prove Theorem~\ref{thm:mainsss}. As mentioned in the introduction, it should be noted that this merely serves as an illustration of how similar problems can be tackled as well since the same method applies to the more general sparse polynomial interpolation problem. However, to avoid a repeat of the analysis of \cite{DBLP:conf/stoc/LokshtanovN10}, we have chosen to restrict ourselves to the \subsetsum problem. Our central contribution over \cite{DBLP:conf/stoc/LokshtanovN10} is that we take advantage of sparsity. 


Given $\vec{a}\in \bbn^n$ and an integer $\modulus \in \bbn$, let $c^\modulus:\bbn^n \rightarrow \bbn^n$ be defined by 
\[
	c^\modulus(\vec{a})_j=\bigg|\bigg\{ X \subseteq [n] : \sum_{e \in X} a_e \equiv j\ (\mod \modulus) \bigg\}\bigg|\text{ \ for every } j\in\bbz_p.
\]
We also use the shorthand $\vec{c(\vec{a})}=\vec{c^{\infty}(\vec{a})}$. We use a corollary from~\cite{DBLP:conf/stoc/LokshtanovN10}:
\begin{cor}[\omitted,\cite{DBLP:conf/stoc/LokshtanovN10}]
\label{cor:ln}
Given an instance $(\vec{a},t)$ of \subsetsum and an integer $\modulus$, $c^\modulus(\vec{a})_t$ can be computed in $\OHS(\modulus)$ time and $\OHS(1)$ space.
\end{cor}

We will also need the following two results on primes:

\begin{thm}[\cite{rosser1941explicit}]
\label{thm:pnt}
If $55 < u$, then the number of prime numbers at most $u$ is at least $\frac{u}{\ln u + 2}$.
\end{thm}

\begin{lem}[\omitted,Folklore]
\label{thm:sam}
There exists an algorithm $\mathtt{pickprime}(u)$ running in $\polylog(u)$ time that, given integer $u \geq 2$ as input, outputs either a prime chosen uniformly at random from the set of primes at most $u$ or $\mathtt{notfound}$. Moreover, the probability that the output is $\mathtt{notfound}$ is at most $\frac{1}{e}$. 
\end{lem}

We will run a data reduction procedure similar to the one of Claim 2.7 in \cite{harnik:1667}, before applying the algorithm of Corollary \ref{cor:ln}. The idea of the data reduction procedure is to work modulo a prime of size roughly $|\supp(c(\vec{a}))|$ or larger:

\begin{lem}
\label{lem:hashsss}
Let $\suppbound \geq |\supp(c(\vec{a}))|$ and let $\beta$ be an upper bound on the number of bits needed to represent the integers, i.e. $2^\beta > \max\{t,\max_{i}a_i\}$. Then, $\prob_\modulus[c(\vec{a})_t=c^\modulus(\vec{a})_t]\geq \frac{1}{2}$, where the probability is taken uniformly over all primes $\modulus \leq \suppbound\beta n(\log \beta)(\log n)$.
\end{lem}
\begin{proof}
Suppose $c(\vec{a})_t \neq c^\modulus(\vec{a})_t$. Then there exists an integer $u \in \supp(c(\vec{a}))$ such that $u \neq t$ and $u \equiv t \pmod\modulus$. This implies that $\modulus$ is a divisor of $|t - u|$, so let us bound the probability of this event. Since $|t - u| \leq 2^\beta n$, it has at most $\beta + \log n$ distinct prime divisors. Let $\gamma=\suppbound\beta n(\log \beta)(\log n)$. It is easy to check that for sufficiently large $\beta$ and $n$:
\[
	\prob_p\bigl[p \text{ divides } |t-u|\bigr] \leq \frac{ \beta + \log n}{\frac{\gamma}{\log \gamma +2}}\leq \frac{ \beta + \log n}{\frac{\gamma}{3(n + \log \beta)}} \leq \frac{1}{2\suppbound},
\]
where we use Theorem \ref{thm:pnt} in the first inequality and $\suppbound \leq 2^n$ in the second inequality. Applying the union bound over the at most $\suppbound$ elements of $\supp(c(\vec{a}))$, the event that there exists a $u \in \supp(c(\vec{a}))$ with $u \neq t$ and $u \equiv t \pmod \modulus$ has probability at most $\frac{1}{2}$.
\end{proof}

Now we give two algorithms utilizing homomorphic hashing, 
one for the case where $\suppbound$ is known, and 
one for the case where $\suppbound$ is not known. 

\begin{thm}
There exists an algorithm that, given an instance $(\vec{a},t)$ of the \subsetsum problem and an integer $\suppbound \geq |\supp(c(\vec{a}))|$ as input, outputs a nonnegative integer $x$ in $\OHS(\suppbound)$ time and polynomial space such that 
(i) $x=0$ implies $c(\vec{a})_t=0$ and
(ii) $\prob[ c(\vec{a})_t=x] \geq \frac{1}{4}$. 
\end{thm}
\begin{proof}
The algorithm is: First, obtain prime $p = \mathtt{pickprime}(\suppbound\beta n(\log \beta)(\log n))$ using Lemma~\ref{thm:sam}. Second, compute and output $c^\modulus(\vec{a})_t$ using Corollary \ref{cor:ln}. Condition (i) holds since $c^\modulus(\vec{a})_t=0$ implies $c(\vec{a})_t=0$ for any $\modulus,t$. Moreover, condition (ii) follows from Lemma~\ref{lem:hashsss} and Lemma~\ref{thm:sam} since $\frac{1}{2}(1-\frac{1}{e}) \geq \frac{1}{4}$. The time and space bounds are met by Corollary $\ref{cor:ln}$ because $p=\OHS(\suppbound)$.
\end{proof}

\begin{proof}[of Theorem~\ref{thm:mainsss}(a)]
The algorithm is the following: Maintain a guess $\suppbound$ of $\supp(c(\vec{a}))$, initially set to $n$. Obtain a prime $p = \mathtt{pickprime}(\suppbound\beta n(\log \beta)(\log n))$ using Lemma~\ref{thm:sam}, and compute $c^\modulus(\vec{a})_t$ using Corollary \ref{cor:ln}. If $c^\modulus(\vec{a})_t=0$ output \no\ since $c(\vec{a})_t=0$; otherwise, attempt to construct a subset $X \subseteq [n]$ such that $\sum_{e \in X}a_e=t$ using self-reduction. If this succeeds, return \yes. Otherwise, double $\suppbound$ and repeat. The expected running time (taken over all primes $p$) is $\OHS(\supp(c(\vec{a})))$ because when $\suppbound \geq \supp(c(\vec{a}))\beta n(\log \beta)(\log n)$ the probability of succesfully constructing a solution or concluding that none exist is at least $\frac{1}{2}$ by the arguments in the proof of Lemma \ref{lem:hashsss}.
\end{proof}

The derandomization for Theorem~\ref{thm:mainsss}(b) is given in the appendix.

\section{Homomorphic Hashing for Linear Satisfiability}
\label{sec:setpart}
In this section we assume that $\fieldF$ is a field of non-even characteristic and that addition and multiplication refer to operations in $\fieldF$. For $s\in\bbn$, we denote by $\mat{\Phi} \in \fieldF^{\bbz_2^s \times \bbz_2^s}$ the \emph{Walsh-Hadamard matrix}, defined for all $\vec{x},\vec{y} \in \bbz_2^s$ by $\Phi_{\vec{x},\vec{y}}=(-1)^{\vec{x}\vec{y}^T}$.

\begin{lem}[Folklore]
	\label{lem:wht}
The Walsh-Hadamard matrix satisfies  $\mat{\Phi}\mat{\Phi} = 2^s\mat{I}$ and, for every $\vec{f},\vec{g} \in \fieldF[\mathbb{Z}_2^s]$, it holds that $(\vec{f} * \vec{g})\mat{\Phi} = \vec{f}\mat{\Phi} \circ \vec{g}\mat{\Phi}$.
\end{lem}
%
%

We first prove the following general theorem, of which Theorem~\ref{thm:mainls}(a) is a special case.
\begin{thm}
\label{thm:z2}
There exists a randomized algorithm that, given as input
\begin{enumerate}
 \item a circuit $C$ with singleton inputs over $\fieldF[\bbz_2^n]$,
 \item an integer $\suppbound \geq |\supp(\vec{v})|$, and 
 \item an element $\vec{t} \in \bbz_2^n$, 
\end{enumerate}
outputs the coefficient $v_{\vec{t}}\in\fieldF$ with probability
at least $\frac{1}{2}$, where $\vec{v} \in \fieldF[\bbz_2^n]$
is the output of $C$. 
The algorithm runs in time $\OHS(\suppbound)$ and uses
$\OHS(\suppbound)$ arithmetic operations in $\fieldF$,
and requires storage for $\OHS(1)$ bits and elements 
of $\fieldF$.
\end{thm}
\begin{algorithm}
\caption{Homomorphic hashing for Theorem \ref{thm:z2}.}
\label{alg:find1}
\begin{algorithmic}[1]
\REQUIRE $\mathtt{hashZ2}$
	\STATE Let $s=\lceil\log \suppbound \rceil$+1. 
	\STATE Choose a matrix $\mat{H} \in \bbz_2^{s \times n}$ uniformly at random from the set of all $s\times n$ matrices with binary entries.
	\STATE Let $h: \fieldF[\bbz_2^n] \rightarrow \fieldF[\bbz_2^s]$ be the homomorphism defined by $h(\vec{a})=\vec{b}$ where $b_{\vec{x}} = \sum_{\vec{y}\in\bbz_2^n:\vec{y}\mat{H}=\vec{x}} a_{\vec{y}}$ for all $\vec{x} \in \bbz_2^s$. Apply $h$ to $C$ to obtain the circuit $C_1$.
	\RETURN $\displaystyle \frac{1}{2^s}\sum_{\vec{x} \in \bbz_2^s}(-1)^{(\vec{t}\mat{H})\vec{x}^T} \ \mathtt{sub}(C_1,\vec{x})$.
\vspace{0.5em}		
\REQUIRE $\mathtt{sub}(C_1,\vec{x})$
		\STATE Let $\varphi:\fieldF[\bbz_2^s]\rightarrow\fieldF$ be the homomorphism defined by $\varphi(\vec{w})=\sum_{\vec{y} \in \mathbb{Z}_2^s}(-1)^{\vec{x} \vec{y}^T}\ w_{\vec{y}}$ for all $\vec{w}\in\fieldF[\bbz_2^s]$. Apply $\varphi$ to $C_1$ to obtain the circuit $C_2$.
		\STATE Evaluate $C_2$ and return the output.
\end{algorithmic}
\end{algorithm}
\begin{proof}
The algorithm is given in Algorithm~\ref{alg:find1}. Let us first analyse the complexity of this algorithm: Steps 1 and 2 can be performed in time polynomial in the input. Step 3 also be done in polynomial time since it amounts to relabeling all input gates with $h(\vec{e})$ where $\vec{e}$ was the old label. Indeed, we know that $\vec{e} \in \fieldF[\bbz_{2}^n]$ is a singleton $\st{v}{\vec{y}}$, so $h(\vec{e})$ is the singleton $\st{v}{\vec{y}\mat{H}}$ and this can be computed in polynomial time. Step 4 takes $\OHS(\suppbound)$ operations and calls to $\mathtt{sub}$, so for the complexity bound it remains to show that each call to $\mathtt{sub}$ runs in polynomial time. Step 5 can be implemented in polynomial time similar to Step 3 since the singleton $\vec{e}=\st{v}{\vec{y}}$ is mapped to $(-1)^{\vec{x}\vec{y}^T}v$. Finally, the direct evaluation of $C_2$ uses $|C_2|$ operations in $\fieldF$. Hence the algorithm meets the time bound, and also the space bound is immediate.

The fact that $\mathtt{hashZ2}$ returns $v_{\vec{t}}$ with probability at least $\frac{1}{2}$ is a direct consequence of the following two claims. Let $\vec{w}$ be the output of $C_1$.
  \begin{clm}[\omitted]\label{clm:one}
    $\prob_{\mat H}[v_{\vec t}=w_{\vec{t}\mat{H}}] \geq \frac{1}{2}$.
  \end{clm}
  
  \begin{clm}[\omitted]\label{clm:two}
    Algorithm $\mathtt{hashZ2}$ returns $w_{\vec{t}\mat{H}}$.
  \end{clm}

\end{proof}

\begin{proof}[of Theorem~\ref{thm:mainls}(a)]
For $1 \leq i \leq n$ and $0\leq w \leq t$ denote by $\vec{A^{(i)}}$ the
$i$th row of $\vec{A}$ and define $\vec f[i,w] \in \mathbb{Q}[\bbz_2^m]$ by
\begin{equation}
\label{eq:setpart}
\vec{f[i,w]}=
\begin{cases}
 \st{1}{\vec{0}}& \text{if } i=w=0,\\
 0& \text{if } i=0 \text{ and } w\neq 0,\\
 \vec f[i-1,w] + \vec f[i-1,j-\omega_i]* \st{1}{\vec{A^{(i)}}} & \text{otherwise.}
\end{cases}
\end{equation}
It is easy to see that for every $1\leq i \leq n$, $0\leq w\leq t$, and $\vec{y} \in \bbz_2^m$, the value $f[i,w]_{\vec{y}}$ is the number of $\vec{x} \in \bbz_2^i$ such that $\tilde{\vec{\omega}}\vec{x}^T =t$ and $\vec{x}\tilde{\mat{A}}=\vec{y}$ where $\tilde{\vec{\omega}}$ and $\tilde{\mat{A}}$ are obtained by truncating $\vec{\omega}$ and $\mat{A}$ to the first $i$ rows. Hence, we let $C$ be the circuit implementing \ref{eq:setpart} and let its output be $\vec v=\sum_{w=0}^{t}\vec{f[n,w]}$. Thus, $v_{\vec{b}}$ is the number of $\vec{x} \in \bbz_2^n$ with $\vec{x}\mat{A}=\vec{b}$ and $\vec{x}\vec{\omega}^T\leq t$.

Also, $|\supp(\vec{v})| \leq 2^{\rank(\vec A)}$ since any element of the support of $\vec{v}$ is a sum of rows of $\vec A$ and hence in the row-space of $\vec A$, which has size at most $2^{\rank(\vec A)}$. To apply Theorem \ref{thm:z2}, let $\fieldF=\mathbb{Q}$ and observe that the computations are in fact carried out over integers bounded in absolute value poly-exponentially in $n$ and hence the operations in the base field can also be executed polynomial in $n$. The theorem follows from Theorem~\ref{thm:z2}.
\end{proof}

To establish Theorem~\ref{thm:mainls}(b), let us first see how to exploit a high linear rank of the matrix $\vec A$ in an instance of \textsc{Linear Sat}. By permuting the rows of $\vec A$ as necessary, we can assume that the first $\rank(\vec A)$ rows of $\vec A$ are linearly independent. We can now partition $\vec x$ into $\vec x=(\vec y,\vec z)$, where $\vec y$ has length $\rank(\vec A)$ and $\vec z$ has length $n-\rank(\vec A)$. There are $2^{n-\rank(\vec A)}$ choices for $\vec z$, each of which by linear independence has at most one corresponding $\vec y$ such that $\vec x\vec A=\vec b$. Given $\vec z$, we can determine the corresponding $\vec y$ (if any) in polynomial time by Gaussian elimination. Thus, we have:


\begin{obs}\label{obs:highrank}
\textsc{Linear Sat} can be solved in $\OHS(2^{n-\rank(\mat{A})})$ time and polynomial space.
\end{obs}

This enables a ``Win/Win approach'' where we distinguish between low and high ranks, and use an appropriate algorithm in each case.

\begin{proof}[of Theorem~\ref{thm:mainls}(b)]
Compute $\rank(\mat{A})$. If $\rank(\mat{A})\geq n/2$, run the algorithm of Observation~\ref{obs:highrank}. Otherwise, run the algorithm implied by Theorem~\ref{thm:mainls}(a).
\end{proof}

\vspace{-2em}
\subsubsection*{Set Partition}

We now give a very similar application to the {\sc Set Partition} problem: given an integer $t$ and a set family $\mathcal{F} \subseteq 2^{U}$ where $|\mathcal{F}|=n$, $|U|=m$, determine whether there is a subfamily $\mathcal{P} \subseteq \mathcal{F}$ with $|\mathcal{P}|\leq t$ such that $\bigcup_{S \in \mathcal{P}}S=U$ and $\sum_{S \in \mathcal{P}}|S|=|U|$. 


The \emph{incidence matrix} of a set system $(U,\mathcal{F})$ is
the $|U| \times |\mathcal{F}|$ matrix $\vec{A}$ whose 
entries $A_{e,S}=[e \in S]$ are indexed by $e\in U$ and $S\in \mathcal{F}$.

\begin{thm}[\omitted]\label{thm:setpartexpspacerank}\label{thm:setpart5}
There exist algorithms that given an instance $(U,\mathcal{F},t)$ of \textsc{Set Partition} output the number of set partitions of size at most $t$ with probability at least $\frac{1}{2}$, and use
	(a) $\OHS(2^{\rank(\vec A)})$ time and polynomial space, and
	(b) $(2^{\rank(\vec A)}+n)m^{\OH(1)}$ time and space,
	where $\mat{A}$ is the incidence matrix of $(U,\mathcal{F})$.
\end{thm}

\section{Homomorphic Hashing for the Union Product}
\label{sec:unionhashing}
\newcommand{\zt}{\boldsymbol{\zeta}}
\newcommand{\mt}{\boldsymbol{\mu}}
\newcommand{\solplus}{\oplus}
\newcommand{\soltimes}{\otimes}
\newcommand{\bigsolplus}{\bigoplus}
\newcommand{\bigsoltimes}{\bigotimes}

\renewcommand{\fieldF}{\mathbb{N}}

In this section our objective is to mimic the approach of the previous section for $\fieldF[(2^U,\cup)]$, where $(2^U,\cup)$ is the semigroup defined by the set union $\cup$ operation on $2^U$, the power set of an $n$-element set $U$. The direct attempt to apply a homomorphic hashing function, unfortunately, fails. Indeed, let $h$ be an arbitrary homomorphism from $(2^U,\cup)$ to $(2^V,\cup)$ with $|V|<|U|$. Let $U=\{e_1,e_2,\ldots,e_n\}$ and consider the minimum value $1\leq j\leq n-1$ with $h(\{e_1,\ldots,e_j\})=\cup_{i=1}^j h(\{e_i\})=\cup_{i=1}^{j+1} h(\{e_i\})=h(\{e_1,\ldots,e_{j+1}\})$; in particular, for $X=\{e_1,\ldots,e_j,e_{j+2},\ldots,e_n\}\neq U$ we have $h(X)=h(U)$, which signals failure since we cannot isolate $X$ from $U$. 

Instead, we use hashing to an algebraic structure based on a poset (the ``Solomon algebra'' of a poset due to~\cite{Solomon1967603}) that is obtained by the technique ``Iterative Compression''. This gives the following main result. For reasons of space we relegate a detailed proof to the appendix; here we will give a simplified version of the proof in the special case of Theorem~\ref{thm:maincnf} in this section.


\begin{thm}[\omitted]
\label{thm:unionhashsmallsupp}
Let and $|U|=n$. There are algorithms that, given a circuit $C$ with singleton inputs in $\fieldF[(2^U,\cup)]$ outputting $\vec{v}$, compute 
\begin{enumerate}[label=(\alph*)]
	\item a list with $v_X$ for every $X \in \supp(\vec{v})$ in $\OHS(|\supp(\vec{v})|^2 \cdot n^{\OH(1)})$ time,
	\item $v_{U}$ in time $\OHS(2^{(1-\alpha/2)n}n^{\OH(1)})$ if $0 < \alpha \leq 1/2$ such that $|\supp(\vec{v})| \leq 2^{(1-\alpha)n}$.
\end{enumerate}
\end{thm}


The above result is stated for simplicity in the unit-cost model, that is, we assume that arithmetic operations on integers take constant time. For the more realistic log-cost model, where such operations are assumed to take time polynomial in the number of bits of the binary representation, we only mention here that our results also hold under some mild technical conditions. Let us first show that Theorem~\ref{thm:maincnf}(a) indeed is a special case of Theorem~\ref{thm:unionhashsmallsupp}:

\begin{proof}[of Theorem~\ref{thm:maincnf}]
Use a circuit over $\mathbb{N}[(2^{[m]},\cup)]$ that implements the expression
\[
 \vec{f} = ( \st{1}{V_1} + \st{1}{\bar V_1} ) * ( \st{1}{V_2} + \st{1}{\bar V_2} )* \ldots * ( \st{1}{V_m} + \st{1}{\bar V_m} ),
\]
where $V_i\subseteq [m]$ (respectively, $\bar V_i \subseteq [m]$) is the set of all indices of clauses that contain a positive (respectively, negative) literal of the variable $v_i$. Then use Theorem~\ref{thm:unionhashsmallsupp} to determine $f_{[m]}$, the number of satisfying assignments of $\phi$.
\end{proof}



Now we proceed with a self-contained proof Theorem~\ref{thm:maincnf}. Given poset $(P,\leq)$, the \emph{M\"obius function} $\mu: P \times P \rightarrow \fieldF$ of $P$ is defined for all $x,y\in P$ by 
\begin{equation}
\label{eq:mob}
	\mu(x,y)=	\begin{cases}
	1 & \text{if } x = y, \\
	-\sum_{x \leq y < z} \mu(y,z) & \text{if } x < z, \\
	0 & \text{otherwise}.
	\end{cases}
\end{equation}
The \emph{zeta transform} $\zt$ and \emph{M\"obius transform} $\mt$ are the $|P| \times |P|$ matrices defined by $\zeta_{x,y}=[x \leq y]$ and $\mu_{x,y}=\mu(x,y)$ for all $x,y\in P$. For a CNF-formula $\phi$ denote $\supp(\phi)$ for the set of all projections of $\phi$. Recall in Theorem~\ref{thm:maincnf} we are given a CNF-Formula $\phi=C^1 \wedge \ldots \wedge C^m$ over $n$ variables. For $i=1,\ldots,m$ define $\phi_i=C_1 \wedge \ldots \wedge C_i$. Then we have the following easy observations

\begin{enumerate}
 \item $\supp(\phi_0)=\{ \emptyset \}$,
 \item $\supp(\phi_i) \subseteq \supp(\phi_{i-1}) \cup \{ X \cup \{i\} : X \in \supp(\phi_{i-1})\}$ for every $i=1,\ldots,m$,
 \item $|\supp(\phi_{i-1})| \leq |\supp(\phi_{i})|$ for every $i=1,\ldots,m$.
\end{enumerate}

Given the above lemma and observations, we will give an algoritm using a technique called \emph{iterative compression} \cite{journals/orl/ReedSV04}. As we will see, by this technique it is sufficient to solve the following ``compression problem'':

\begin{lem}\label{lem:compresssimple}
Given a CNF-formula $\phi=C_1 \wedge \ldots \wedge C_m$ and a set family $\mathcal{F} \subseteq 2^{[m]}$ with $\supp(\phi) \subseteq \mathcal{F}$, the set $\supp(\phi)$ can be constructed in $\OHS(|\mathcal{F}|^2)$ time.  
\end{lem}
\begin{proof}
In what follows $\vec{a}\in\{0,1\}^n$ refers to an assignment of 
values to the $n$ variables in $\phi$.
Define $\vec{f} \in \mathbb{N}^{2^{[m]}}$ for all $X\subseteq [m]$ by
	\[
		f_X = |\{ \vec{a} \in \{0,1\}^n : \forall i \in [m] \text{ it holds that }  \vec{a} \text{ satisfies } C_i \text{ iff } i \in X \}|.
	\]
It is easy to see that $\supp(\vec{f})=\supp(\phi)$, so if we know $f_X$ for every $X \in \mathcal{F}$ we can construct $\supp(\phi)$ in $|\mathcal{F}|$ time. Towards this end, first note that for every $Y \subseteq [m]$ it holds that
\begin{align*}
(\vec{f} \zt)_Y &=\sum_{\substack{X \in \supp(\vec{f}) \\ X \subseteq Y}}f(X)= \sum_{X \subseteq Y} f(X)\\
&\ = |\{ \vec{a} \in \{0,1\}^n : \forall i \in [m] \text{ it holds that } \vec{a} \text{ satisfies } C_i \text{ only if } i \in Y \}|.
\end{align*}
Second, note that the last quantity can be computed in polynomial time: since every clause outside $Y$ must not be satisfied, each such clause forces the variables that occur in it to unique values; any other variables may be assigned to arbitrary values. That is, the count is 0 if the clauses outside $Y$ force at least one variable to conflicting values, otherwise the count is $2^a$ where $a$ is the number of variables that occur in none of the clauses outside $Y$. 

Now the algorithm is the following: for every $X \in \mathcal{F}$ compute $(\vec{f}\zt)_X$ in polynomial time as discussed above. Then we can use algorithm $\mathtt{mobius}$ as described below to obtain $f_X$ for every $X \in \mathcal{F}$ since it follows that $\vec{f}=\mathtt{mobius}((\mathcal{F},\subseteq),\vec{f}\zt)$ from the definition of $\mt$ and the fact that $\mt\zt=\mat{I}$. Algorithm $\mathtt{mobius}$ clearly runs in $\OHS(|P|^2)$ time, so this procedure meets the claimed time bound.
\vspace{-1em}
\begin{algorithm}[H]
\begin{algorithmic}[1]
\REQUIRE $\mathtt{mobius}((P,\leq),\vec{w})$
		\STATE Let $P=\{v_1,v_2\ldots,v_{|P|}\}$ such that $v_i \leq v_i$ implies $i \geq j$.
		\STATE $\vec{z} \leftarrow \vec{w}$.
		\FOR{$i= 1,2,\ldots,|P|$}
			\FOR{every $v_j \leq v_i$}
				\STATE $z_i = z_i - z_j$
			\ENDFOR
		\ENDFOR
		\RETURN $\vec{z}$.
\end{algorithmic}
\end{algorithm}
\vspace{-2.5em}
\end{proof}

\begin{proof}[of Theorem~\ref{thm:maincnf}, self-contained]
 Recall that we already know that $\supp(\phi_0)=\{\emptyset\}$. Now, for $i = 1,\ldots,m$ we set $\mathcal{F} = \supp(\phi_{i-1}) \cup \{ X \cup \{i\} : X \in \supp(\phi_{i-1})\}$ and use $\mathcal{F}$ to obtain $\supp(\phi_{i-1})$ using Lemma~\ref{lem:compresssimple}. In the end we are given $\supp(\phi_m)$ and since $\phi_m$ is exactly the original formula, the input is a yes-instance if and only if $[m] \in \supp(\phi_m)$. The claimed running time follows from Observations 1 and 3 above and the running time of algorithm $\mathtt{mobius}$.
\end{proof}

\subsubsection{Set Cover}

We will now give an application of Theorem~\ref{thm:unionhashsmallsupp}(b) to {\sc Set Cover}: Given a set family $\mathcal{F} \subseteq 2^U$ where $|U|=n$ and an integer $k$, find a subfamily $\mathcal{C} \subseteq \mathcal{F}$ such that $|\mathcal{C}|= k$ and $\bigcup_{S \in \mathcal{C}}S=U$.

\begin{thm} Given an instance of {\sc Set Cover}, let $0 < \alpha \leq 1/2$ be the largest real such that $|\{\bigcup_{S \in C}S: \mathcal{C} \subseteq \mathcal{F} \wedge |\mathcal{C}|=k \}| \leq 2^{(1-\alpha)n}$. Then the instance can be solved in $\OHS(2^{(1 - \alpha/2)n}n^{\OH(1)})$ time (and exponential space).
\end{thm}
\begin{proof}
Let $\mathcal{F}=\{S_1,\ldots,S_m\}$, and for every $1\leq i \leq m$ and $1 \leq j \leq n$ define $\vec{f[i,j]}$ as follows:
\[
\vec{f[i,j]}=
\begin{cases}
 \st{1}{\emptyset}& \text{if } i=j=0,\\
 \vec f[i-1,j] + \vec f[i-1,j-1]* \st{1}{S_i} & \text{otherwise.}
\end{cases}
\]
It is easy to see that for every $X \subseteq U$ we have that $f[i,j]_X$ is the number of $\mathcal{C} \subseteq \{S_1,\ldots,S_i\}$ such that $|\mathcal{C}|=j$ and $\bigcup_{S \in \mathcal{C}}S=X$. Hence the theorem follows directly by applying Theorem~\ref{thm:unionhashsmallsupp}(b) by interpreting the above recurrence as a circuit in order to determine the value $f[m,k]_U$.
\end{proof}




\bibliographystyle{abbrv}
{\bibliography{hashing}}

\appendix

\clearpage

\begin{center}
{\LARGE\bf APPENDIX}
\end{center}

\section{Omitted proofs}

\subsection{Proofs omitted in Section~\ref{sec:sss}}

\subsubsection{Proof of Corollary~\ref{cor:ln}}\label{app:proofs}

For the proof we use the following result as a blackbox:
\begin{thm}[\cite{DBLP:conf/stoc/LokshtanovN10}]
\label{thm:ln}
Given an instance $(\vec{a},t)$ of \subsetsum,
$c(\vec{a})_t$ can be computed in $\OHS(t)$ time and $\OHS(1)$ space.
\end{thm}

\begin{proof}[of Corollary~\ref{cor:ln}]
By reducing modulo $p$ we can assume that $a_1,\ldots,a_n,t < \modulus$. We have
\[
	c^p(\vec{a})_t = \sum_{i \equiv t (\mod p)}^{n} c(\vec{a})_{i} =  \sum_{j=0}^{n} c(\vec{a})_{t+jp},
\]
where the first equality holds by definition and the second equality follows from the fact that $c(\vec{a})_{i}=0$ when $i > np$ or $0<i$. The latter expression can be evaluated in the claimed resource bounds using Theorem \ref{thm:ln}.
\end{proof}

\subsubsection{Proof of Theorem~\ref{thm:sam}}
Take an integer $i \leq u$ uniformly at random, and check whether it is a prime using the polynomial time algorithm of \cite{agrawal2004primes}. If $i$ is prime, then output $i$ and halt; otherwise repeat. If no prime is found after $\ln(u)+2$ repetitions, output $\mathtt{notfound}$ and halt. The upper bound on the probability of failure follows from Theorem \ref{thm:pnt} since the probability that $\mathtt{notfound}$ is returned is at most $(1-(\ln(u)+2)^{-1})^{\ln(u)+2} \leq \frac{1}{e}$.

\subsubsection{Proof of Theorem~\ref{thm:mainsss}(b))}
Let $\beta = \lceil \log \max\{t,\max_{i}a_i\} \rceil$. Call a prime $p$ \emph{bad} if $c(\vec{a})_t \neq c^\modulus(\vec{a})_t$. Analoguously to the proof of Lemma \ref{lem:hashsss}, there are at most $(\beta + \lceil\log n\rceil)\suppbound$ bad primes. Let $p_1,\ldots,p_l$ be the first $l$ prime numbers in increasing order where $l= 2(\beta + \lceil\log n\rceil)\suppbound+1$. Since there are only $(\beta + \lceil\log n\rceil)S$ bad primes, the majority of the set of integers $\{c^{p_i}(\vec{a})_t\}_{1\leq i \leq l}$ will be equal to $c(\vec{a})_t$. Then use the folklore Majority voting algorithm~\cite{DBLP:conf/birthday/Moore91} as in Algorithm~\ref{alg:prfapp}:

\begin{algorithm}
\label{alg:prfapp}
\begin{algorithmic}[1] 
\STATE $M,C,i,j \leftarrow 0$.
\WHILE{$i \leq l$}
	\STATE $j \leftarrow j + 1$
	\IF{the deterministic primality testing from \cite{agrawal2004primes} returns that $j$ is prime}
	\STATE $i \leftarrow i + 1$; $p_i\leftarrow j$
		\STATE Compute $c^{p_i}(\vec{a})_t$ using Corollary \ref{cor:ln}.
		\IF{$M=0$}
			\STATE $C \leftarrow c^{p_i}(\vec{a})_t$
		\ELSIF{$c^{p_i}(\vec{a})_t=C$}
			\STATE  $M \leftarrow M+1$
		\ELSE 
	 	\STATE $M \leftarrow M-1$
		\ENDIF
	\ENDIF
\ENDWHILE
\RETURN $C$
\end{algorithmic}
\end{algorithm}

The correctness follows from the above discussion and the correctness of the majority voting algorithm~\cite{DBLP:conf/birthday/Moore91} that is folklore and implemented in the algorithm. For the running time, note that $p_l$ is $\OHS(\suppbound)$ by Theorem \ref{thm:pnt} and hence the running time is $\OHS(\suppbound^2)$ since Step 6 takes time $\OHS(\suppbound)$. It is clear that the algorithm can be implemented using polynomial space.

\subsection{Proofs omitted in Section~\ref{sec:setpart}}
\subsubsection{Proof of Claim~\ref{clm:one}}
  	For every $\vec{a},\vec{b}\in \fieldF[\bbz_2^n]$ we have $(\vec{a}+\vec{b})\mat{H}=\vec{a}\mat{H} + \vec{b}\mat{H}$ and hence $h$ is easily seen to be a homomorphism by Observation~\ref{obs:hom}. Thus, by Observation \ref{obs:homcir} we know that $\vec{w}=h(\vec{v})$, that is, for every $\vec{z} \in \bbz_2^s$
  	\[
  		\vec w_{\vec{z}} = \sum_{\vec{y}\in\bbz_2^n:\vec{y}\mat{H}=\vec{z}}\vec v_{\vec{y}}.
  	\]
  	Hence, if $\vec v_{\vec t}\neq\vec w_{\vec{t}\mat{H}}$, there must exist $\vec{y} \in \supp(\vec{v})$ such that $\vec{y} \neq \vec{t}$ and $\vec{y}\mat{H}=\vec{t}\mat{H}$. Equivalently, $(\vec{y}-\vec{t})\mat{H}=\vec{0}$. For any $\vec x\in\bbz_2^n$ with $\vec{x}\neq \vec{0}$ we have
  	\[
  		\prob_{\mat{H}}[\vec{x}\mat{H} = 0]=\prod_{i=1}^s\prob_{\mat{H}}[(\vec{x}\mat{H})_i = 0]=2^{-s},
  	\]
  	where the probability is taken uniform over all binary $s \times n$ matrices, and the two equalities follow from the fact that the random variables $(\vec{x}\mat{H})_i$ for $i=1,\ldots,s$ are independent and uniformly distributed. Now the claim follows by taking the union bound over all elements in the support:
  	\[
  		\prob[\vec v_{\vec t}\neq\vec w_{\vec{t}\mat{H}}] \leq \prob[\exists \vec{x} \in \supp(\vec{v}):\ \vec{x}\neq \vec{t} \wedge \vec{x}\mat{H}=\vec{t}\mat{H}] \leq |S|2^{-s} \leq \frac{1}{2}.
  	\]

\subsubsection{Proof of Claim~\ref{clm:two}}
   Let $\vec{b} \in \fieldF[\bbz_2^s]$ such that $\vec{b}_{\vec{x}}=\mathtt{sub}(C_1,\vec{x})$ for every $\vec{x} \in \bbz_2^s$. It suffices to show that $\vec{b}=\vec{w}\mat{\Phi}$ since $\mathtt{hashZ2}$ returns $\frac{1}{2^s}(\vec{b}\mat{\Phi})_{\vec{t}\mat{H}}$ as can bee seen from Line 4, and this is equal to $\vec w_{\mat H\vec t}$ by Lemma \ref{lem:wht}. For proving that $\vec{b}=\vec{w}\mat{\Phi}$, we first claim that $\varphi$ is a homomorphism from $\fieldF[\bbz_2^s]$ to $\fieldF$ since, for $\vec{a},\vec{b} \in \fieldF[\bbz_2^s]$, we have that $\varphi(\vec{a}+\vec{b})$ equals
\[
 \sum_{\vec{y} \in \bbz_{2}^s}(-1)^{\vec{x}\vec{y}^T}(a_{\vec{y}}+b_{\vec{y}})= \sum_{\vec{y} \in \bbz_{2}^s}(-1)^{\vec{x}\vec{y}^T}a_{\vec{y}} + \sum_{\vec{y} \in \bbz_{2}^s}(-1)^{\vec{x}\vec{y}^T}b_{\vec{y}}= \varphi(\vec{a})+\varphi(\vec{b}),
\]
and $\varphi(\vec{a} * \vec{b})$ equals
\[ 
\sum_{\vec{y} \in \bbz_2^s}(-1)^{{\vec{x} \vec{y}^T}} \sum_{\vec{y_1}+\vec{y_2}=\vec{y}}a_{\vec{y_1}}b_{\vec{y_2}} = \sum_{\vec{y_1} \in \bbz_2^s}(-1)^{\vec{x}\vec{y_1}^T}a_{\vec{y_1}} \sum_{\vec{y_2} \in \bbz_2^s}(-1)^{\vec{x}\vec{y_2}^T}b_{\vec{y_2}} = \varphi(\vec a) \varphi(\vec b).
\]
Then, by Observation \ref{obs:homcir}, $\mathtt{sub}(C_1,\vec{x})$ returns $\varphi(\vec w)$. For $\vec{a} \in \bbz[\bbz_2^s]$ we have $\varphi(\vec{a})=(\vec{a}\mat{\Phi})_{\vec{x}}$ so $\mathtt{sub}(C_1,\vec{x})=(\vec{w}\mat{\Phi})_{\vec{x}}$ and hence $\vec{b}=\vec{w}\mat{\Phi}$.

\subsubsection{Proof of Theorem~\ref{thm:setpart5}}

\begin{proof}[of Theorem~\ref{thm:setpart5}.a)]
Use Theorem~\ref{thm:mainls}. Assume $\mathcal{F}=\{S_1,\ldots,S_n\}$ and create the instance $(\mat{A},\vec{b},\vec{\omega},t')$ of \textsc{Linear Sat} where $\mat{A}$ is the incidence matrix of the set system $(\mathcal{F},U)$, $\vec{b}=1$, $\omega_i=|S_i|n+1$ and $t'=nm+t$. It is easy to see that the algorithm of Theorem~\ref{thm:mainls} returns exactly the number of set partitions of size at most $t$
\end{proof}
For the second part we will need the following folklore result.
\begin{thm}[Fast Walsh-Hadamard transform, Folklore]
\label{thm:fastwht}
Given a vector $\vec{a} \in \fieldF[\bbz_2^s]$, $\vec{a}\mat{\Phi}$ can be computed in time $\mathcal{O}(2^ss)$ and using $\mathcal{O}(2^ss)$ operations in $\fieldF$.
\end{thm}
\begin{proof}[of Theorem~\ref{thm:setpart5}.b)]
Consider the following circuit $C$ over $\mathbb{Q}[\bbz_2^m]$:
\begin{align}
\label{eq:expspace1}
f[i,j]&=
\begin{cases}
 \displaystyle \st{1}{\vec{0}}& \text{if } i=j=0\\
 \displaystyle 0& \text{if } i=0\text{ and }j\neq0\\
 \displaystyle \sum_{h=0}^j f[i-1,h]g[j-h] & \text{otherwise, where} 
\end{cases}\\
\label{eq:expspace2}
	g[j] &= \sum_{i=1}^n \big[|S_i|=j\big] \st{1}{S_i}
\end{align}

For every $\vec{x}\in\bbz_2^m$ and non-negative integers $i$ and $j$, the coefficient $f[i,j]_{\vec{x}}$ counts the number of ways to choose an $i$-tuple of sets in $\mathcal{S}$ such that their sizes sum up to $j$ and their characteristic vectors sum to $\vec{x}$ in $\bbz_{2}^m$. Thus $\supp(f[t,m]) \leq 2^{\rank(\mat{A})}$. Furthermore, $f[t,m]_{\vec{1}}$ is the number of set partitions of size $t$ times $t!$. Indeed, if a $t$-tuple of sets from $\mathcal{S}$ contributes to $f[t,n]_{\vec{1}}$, each element of $U$ must occur in a unique set in the $t$-tuple. It remains to compute $(f[t,n])_{\vec{1}}$. For this we will use algorithm $\mathtt{hashZ2}$ with $s=\rank(A)$, except that we replace Line 4 with the following to compute $\vec w_{\vec t\mat H}$, where $\vec{w}$ is the output of $C_1$:

\begin{algorithm}[H]
\caption{Changes to Algorithm~\ref{alg:find1} to implement Theorem~\ref{thm:setpartexpspacerank}.}
\label{alg:findexpspace}
\begin{algorithmic}[1]
\setcounter{ALC@line}{3}
\FOR{every $0 \leq j \leq m$}
	\STATE Compute and store $h(\vec{g[j]})$ using \eqref{eq:expspace2}
	\STATE Compute and store $h(\vec{g[j]})\mat{\Phi}$ using Theorem \ref{thm:fastwht}
\ENDFOR
\RETURN $\displaystyle \frac{1}{2^s}\sum_{\vec{x} \in \bbz_2^s}(-1)^{(\vec{1}\mat{H})\vec{x}^T} (\vec{w}\mat{\Phi})_{\vec{x}}$, using \eqref{eq:expspace1} and the stored values to compute the vector $\vec{w}\mat{\Phi}$. 
\end{algorithmic}
\end{algorithm}
The correctness follows from Claim \ref{clm:one} and the observation that the inversion formula from Theorem \ref{lem:wht} is returned on Line 7. Indeed, $h$ is a homomorphism and $\Phi$ is a bijective homomorphism, so 
 Observation \ref{obs:homcir} enables us to compute 
$\vec{w}\mat{\Phi}$ using \eqref{eq:expspace2}.

To establish the time and space complexity, we observe that Steps 5 and 6 take $(2^s+n)m^{\mathcal{O}(1)}$ time by elementary analysis and Theorem~\ref{thm:fastwht}, and that Step 7 takes $2^sm^{\mathcal{O}(1)}$ time since we can compute $\vec{w}=h(\vec{f[t,m]})$ via~\eqref{eq:expspace1} in $\OH(m^3)$ operations in $\mathbb{Q}^{\bbz_2^s}$ by relying on the stored values $h(\vec g[j])\mat{\Phi}$, where each operation requires $\OH(2^s s)$ time by Theorem~\ref{thm:fastwht}.
\end{proof}

\subsection{Proofs of Section~\ref{sec:unionhashing}}

This section is dedicated to the proof of Theorem~\ref{thm:unionhashsmallsupp}. Instead of a combinatorial proof similar to the one of Section~\ref{sec:unionhashing}, we use the algebraic perspective in this proof since we feel it gives a more fundamental insight into the hashing function used. To obtain Theorem~\ref{thm:unionhashsmallsupp}, we prove the following generalization of Lemma~\ref{lem:compresssimple}.

\begin{lem}
\label{thm:unionhashing}
There is an algorithm that, given circuit $C$ over $\fieldF[(2^U,\cup)]$ with singleton inputs outputting $\vec{v}$, and a set family $\mathcal{F} \subseteq 2^U$ such that $\supp(\vec{v}) \subseteq \mathcal{F}$ computes a list with $v_X$ for every $X \in \supp(\vec{v})$ in $\mathcal{O}^*(|\mathcal{F}|^2)$ time.
\end{lem}

\subsubsection{From Lemma~\ref{thm:unionhashing} to Theorem~\ref{thm:unionhashsmallsupp}.}\label{sec:fromlemtothmunion}

We now combine Lemma~\ref{thm:unionhashing} with the iterative compression technique to obtain Theorem~\ref{thm:unionhashsmallsupp} in a way very similar to Section~\ref{sec:unionhashing}. Instead of the subformula $\phi_i$ from Section~\ref{sec:unionhashing}, we need the following notion.

\begin{defi}\label{defi:restri}
Given a circuit $C$ over $\fieldF[(2^U,\cup)]$ and $X \subseteq U$, the \emph{restriction of $C$ to $X$} is the circuit obtained by applying the function $h$ to it, where, for $Y \subseteq X$
\[
	h(\vec{v})_Y = \sum_{W \subseteq U: W \cap X = Y}v_W.
\]
\end{defi}
\begin{proof}[of Theorem~\ref{thm:unionhashsmallsupp}]
Assume $U=\{e_1,e_2,\ldots,e_n\}$, and for every $i=0,\ldots,n-1$, let $C^{i}$ be the restriction of $C$ to $\{e_1,\ldots,e_i\}$ and $\vec{v^i}$ be the output of $C^i$. Then, since taking restrictions is homomorphic, we have by Observation~\ref{obs:homcir} that
\[
v^i_Y = \sum_{W \subseteq U: W \cap \{e_1,\ldots,e_i\}=Y}v_W \leq \sum_{W \subseteq U: W \cap \{e_1,\ldots,e_{i-1}\}=Y \setminus \{e_i\}}v_W=v^{i-1}_{Y \setminus \{e_i\}}.
\]
This implies that if $X \in \supp(\vec{v^i})$, then $X \setminus \{e_i\} \in \supp(\vec{v^{i-1}})$ and hence
\begin{equation}\label{eq:suppgrow}
 \supp(\vec{v^i}) \subseteq \supp(\vec{v^{i-1}}) \cup \{ X \cup \{e_i\} : X \in \supp(\vec{v^{i-1}})\}.
\end{equation}
We now describe how to implement the theorem. First note that $\supp(\vec{v^0})={\emptyset}$. For $i=1,\ldots,n$ do the following: set $\mathcal{F}^i=\supp(\vec{v^{i-1}}) \cup \{ X \cup \{e_i\} : X \in \supp(\vec{v^{i-1}}) \}$. Then by~\eqref{eq:suppgrow}, $\supp(\vec{v^i}) \subseteq \mathcal{F}^i$. Hence, by invoking the algorithm of Lemma~\ref{thm:unionhashing} using $C^i$ as the circuit and $\mathcal{F}^i$ as the set family, we can obtain $v^i_X$ for every $X \in \supp(\vec{v^i})$ in $\mathcal{O}^*(|\mathcal{F}^i|^2)$ time which is $\mathcal{O}^{*}((2 |\supp(v^{i-1})|)^2)$ time. From these values we can easily obtain the support of $\vec{v^i}$. After $n$ steps, we have computed $v_X=v^n_X$ for every $X \in \supp(\vec{v})$. The claimed running time follows from the fact that $|\supp(\vec{v^{i-1}})| \leq |\supp(\vec{v^{i}})|$ for every $i=1,\ldots,n$.
\end{proof}

The remainder of this section is devoted to the proof of Lemma~\ref{thm:unionhashing}. The idea is to use the set family $\mathcal{F}$ to create a poset and homomorphically hash $\fieldF[(2^U,\cup)]$ to the so-called Solomon algebra of the poset. We will first recall all necessary notions and properties, second introduce the hash function, and third give the algorithm implementing Lemma~\ref{thm:unionhashing}.

\subsubsection{Preliminaries on M\"obius inversion and the Solomon algebra}
\label{sec:mobsol}
Let $(P,\leq)$ be a poset. An element $c\in P$ is said to be an {\em upper bound} of a set $X \subseteq P$ if $x\leq c$ holds for all $x \in X$. The set $X$ is said to have a \emph{join} in $P$ if there exists an upper bound $c \in P$ (called \emph{the join}) of $X$ such that, for all upper bounds $d$ of $X$, it holds that $c \leq d$. For $X \subseteq P$, let us write $\bigvee X$ for the join of $X$; for the join of $X=\{a,b\}$ we write simply $a \vee b$. The {\em open interval} $(x,y)$ is the poset induced by the set $\{e \in P : x < e < y \} \subseteq P$. We write $[x,y)$, $(x,y]$ and $[x,y]$ for the analogous (half-)closed interval. 


Let $P$ be a poset. The \emph{M\"obius function} $\mu: P \times P \rightarrow \fieldF$ of $P$ is defined for all $x,y\in P$ by 
\begin{equation}
	\mu(x,y)=	\begin{cases}
	1 & \text{if } x = y, \\
	-\sum_{x \leq z < y} \mu(x,z) & \text{if } x < y, \\
	0 & \text{otherwise}.
	\end{cases}
\end{equation}
The \emph{zeta transform} $\zt$ and \emph{M\"obius transform} $\mt$ are the $|P| \times |P|$ matrices defined by $\zeta_{x,y}=[x \leq y]$ and $\mu_{x,y}=\mu(x,y)$ for all $x,y\in P$.


The following combinatorial interpretation is particularly useful:
\begin{thm}[Hall (see \cite{stanleyenumerative},Proposition 3.8.5)]\label{thm:hall}
For all $x,y\in P$,  $\mu(x,y)$ is the number of even chains in $(x,y)$ minus the number of odd chains in $(x,y)$.
\end{thm}
We include a proof for completeness
\begin{proof}
Use induction on the number of elements in the interval $(x,y)$. If $x=y$, $(x,y)$ contains only the empty chain, which is even. If $x \leq y$, group all chains on their smallest element $y$. The contribution of all these chains is exactly $-\mu(y,z)$ since odd chains are extended to even chains and vice-versa by adding $y$.
\end{proof}

It is known that $\mt$ and $\zt$ are mutual inverses. To see this, note that $(\mt \zeta)_{xy}$ can be interpreted as the number of even chains minus the number of odd chains in the interval $(x,y]$ which can be seen to be $0$ if $x \neq y$ by a pairing argument, and $1$ otherwise. This principle is called \emph{M{\"o}bius inversion}.


\begin{defi}[\cite{Solomon1967603}]
Let $(P,\leq)$ be a poset. The \emph{Solomon algebra} $\fieldF[P]$ is the set $\fieldF^P$ equipped with coordinate-wise addition $\solplus$ and the \emph{Solomon product} $\soltimes$ defined for all $\vec f,\vec g\in\fieldF[P]$ and $z\in P$ by
\[
	(\vec{f} \solplus \vec{g})_{z} = f_z+g_z \qquad (\vec{f} \soltimes \vec{g})_z = \sum_{x,y \in P} \left(\sum_{x,y \leq q \leq z} \mu(q,z)\right) f_x g_y
\]
\end{defi}

The following properties of the Solomon algebra will be useful:

\begin{lem}
\label{lem:assoc}
For every $\vec{v^{1}},\ldots\vec{v^{k}} \in \fieldF[P]$ and $s \in P$ we have
\[ \left( \bigsoltimes_{i=1}^n \vec{v^i} \right)_s = \sum_{a_1,\ldots,a_k \in P} \left( \sum_{a_1,\ldots,a_k \leq r \leq s} \mu(r,s) \right) \prod_{i=1}^k v^i_{a_i} \]
\end{lem}
\begin{proof}
Use induction of $n$. For $n=1$, the statement clearly holds. Otherwise, we have by definition of $\soltimes$ that
\begin{align*}
\left( \bigsoltimes_{i=1}^k \vec{v^i} \right)_s&= \sum_{w,a_k \in P} \left( \sum_{w,a_k \leq r \leq s} \mu(r,s) \right) \left( \bigsoltimes_{i=1}^{k-1} \vec{v^i} \right)_w v^k_{a_k}\\
																							 &\{ \text{Induction Hypothesis} \}\\
																							 &= \sum_{w,a_k \in P} \left( \sum_{w,a_k \leq r \leq s} \mu(r,s) \right) \sum_{a_1,\ldots,a_{k-1} \in P} \left( \sum_{a_1,\ldots,a_{k-1} \leq q \leq w} \mu(q,w) \right) \prod_{i=1}^{k} v^i_{a_i}\\
																							 &\{ \text{Reordering summations} \}\\
																							 &= \sum_{a_1,\ldots,a_k,w \in P} \left( \sum_{\substack{w,a_k \leq r \leq s \\ a_1,\ldots,a_{k-1} \leq q \leq w}} \mu(r,s)\mu(q,w) \right) \prod_{i=1}^{k} v^i_{a_i}\\
																							 &\{ \text{Reordering summations} \}\\
																							 &= \sum_{a_1,\ldots,a_k \in P} \sum_{\substack{ a_1,\ldots,a_{k-1} \leq q \\ a_k \leq r }} \left( \sum_{q\leq w \leq r} \mu(q,w) \right) \mu(r,s) \prod_{i=1}^{k} v^i_{a_i}\\
																							 &\{ \text{M{\"o}bius inversion} \}\\
																							 &= \sum_{a_1,\ldots,a_k \in P} \sum_{\substack{a_1,\ldots,a_{k-1} \leq q \\ a_k \leq r }}[q=r] \mu(r,s) \prod_{i=1}^{k} v^i_{a_i}\\
																							 &= \sum_{a_1,\ldots,a_k \in P} \left( \sum_{a_1,\ldots,a_{k} \leq q} \mu(q,s) \right) \prod_{i=1}^{k} v^i_{a_i}																							 
\end{align*}
\end{proof}

\begin{thm}[\cite{Solomon1967603}]
\label{thm:zethom}
The zeta transform is a homomorphism from $\fieldF[P]$ to $\fieldF^P$; that is, for every $\vec{f},\vec{g}\in\fieldF[P]$ it holds that $(\vec{f} \solplus \vec{g})\zt=\vec{f}\zt+\vec{g}\zt$ and $(\vec{f} \soltimes \vec{g})\zt=\vec{f}\zt\circ\vec{g}\zt$.
\end{thm}
\begin{proof}
For every $w \in P$ we have
\begin{align*}
	((\vec{f}\solplus\vec{g})\zt)_{w} &= \sum_{z \leq w}(f_z + g_z)= \sum_{z \leq w}f_z + \sum_{z \leq w}g_z = (\vec{f}\zt)_w + (\vec{g}\zt)_w,\quad\text{and}\\
	((\vec{f}\soltimes\vec{g})\zt)_{w} &= \sum_{z \leq w} \sum_{x,y \in P} \left( \sum_{x,y \leq q \leq z} \mu(q,z) \right)f_x g_y \\
														 &= \sum_{x,y,q,z \in P} [x,y \leq q \leq z \leq w] \mu(q,z) f_x g_y \\
														 &= \sum_{x,y,q \in P} [x,y \leq q \leq w] \left(\sum_{q \leq z \leq w} \mu(q,z)\right) f_x g_y \\
														 &= \sum_{x,y,q \in P} [x,y \leq q \leq w] [q=w] f_x g_y \\
														 &= \sum_{x,y \in P} [x,y \leq w] f_x g_y = \left(\sum_{x \leq w}f_x\right) \left( \sum_{y \leq w}  g_y \right) = (\vec{f}\zt)_{w} (\vec{g}\zt)_{w}.
\end{align*}
\end{proof}

\begin{lem}[\cite{Solomon1967603}]
\label{lem:solring}
Let $P$ be a poset with a minimum element $\hat{0}$. Then, $\fieldF[P]$ is a commutative ring with the multiplicative identity $\st{1}{\hat{0}}$.
\end{lem}
\begin{proof}
The singleton $\st{1}{\hat{0}}$ is the multiplicative identity because
\[
	(\st{1}{\hat{0}} \soltimes \vec{g})_{z}=\sum_{ y \leq z } \left( \sum_{y \leq q \leq z}\mu(q,z) \right)g_y=g_z,
\]
where the last equality follows from M\"obius inversion. To see that $\fieldF[P]$ is a commutative ring, note that $\zt$ is an isomorphism from $\fieldF[P]$ to the ring $\fieldF^P$. Indeed, Theorem \ref{thm:zethom} shows that $\zt$ is a homomorphism, and $\zt\mt=\mt\zt=\mat I$ shows that $\zt$ is bijective.
\end{proof}

\begin{lem}
\label{lem:idempot}
For singletons $\st{c}{x} , \st{d}{x} \in \fieldF[P]$ it holds that $\st{c}{x}\soltimes \st{d}{x}=\st{cd}{x}$.
\end{lem}
\begin{proof}
For all $z\in P$, we have that $(\st{c}{x} \soltimes \st{d}{x})_{z}$ equals
\[
  \sum_{w,y \in P} \left( \sum_{w,y \leq q \leq z} \mu(q,z) \right) [w=x]c[y=x]d = \sum_{x \leq q \leq z} \mu(q,z)cd=[x=z]cd.
\]
\end{proof}

\subsubsection{The hash function for Lemma~\ref{thm:unionhashing}} Let $U$ be an $n$-element set and let $(P,\leq)$ be a poset that satisfies $U \subseteq P$ and has a minimum element $\hat{0}$. Define the function $h: \fieldF[(2^U,\cup)]\rightarrow \fieldF[P]$ by setting, for all $\vec{a}\in\fieldF[(2^U,\cup)]$, 
\begin{equation}
\label{eq:hdef}
 	h(\vec{a}) = \bigsolplus_{X \subseteq U} \st{a_X}{\hat{0}} \soltimes \bigsoltimes_{e \in X} \st{1}{e}.
\end{equation}

The following two lemmas combined show that the function $h$ is in fact an homomorphic hash function:

\begin{lem}\label{lem:homhashunion}
If $(P,\leq)$ is a poset with minimum element, $U \subseteq P$ and $\vec{v} \in \fieldF[(2^U,\cup)]$ such that every $X \in \supp(\vec{v})$ has a join in $P$ and for every $X,Y \in \supp(\vec{v})$,
\begin{equation}\label{eqn:reqprop}
	\vee \{\{e\} : e \in X\} \neq \vee \{\{e\} : e \in Y\},
\end{equation}
then for every $X \in \supp(v)$, $h(\vec{v})_{\vee \{\{e\} : e \in X\}}=v_X$.
\end{lem}
\begin{proof}
Denote $y = \vee \{\{e\} : e \in X\}$. Then,
\begin{align*}
	(h(\vec{v}))_{y} &= \biggl(\bigsolplus_{X \subseteq U} \st{v_X}{\hat{0}} \soltimes \bigsoltimes_{e \in X} \st{1}{e}\biggr)_{y} \\
												 &\{ \text{for }X \notin \supp(\vec{v}) \text{ we have } v_X=0 \}\\
												 &= \biggl(\bigsolplus_{X \in \supp(\vec{v})} \st{v_X}{\hat{0}} \soltimes \bigsoltimes_{e \in X} \st{1}{e}\biggr)_{y}\\
												 &\{ \text{Definition of $\solplus$} \}\\
												 &= \sum_{X \in \supp(\vec{v})} \bigl(\st{v_X}{\hat{0}} \soltimes \bigsoltimes_{e \in X} \st{1}{e}\bigr)_{y}\\
												 &\{ \text{Applying Lemma \ref{lem:assoc}, denoting } X=\{e_1,\ldots,e_{|X|}\} \}\\
												 &= \sum_{X \in \supp(\vec{v})} \sum_{a_0,\ldots,a_{|X|} \in P} \left( \sum_{a_0,\ldots,a_{|X|} \leq q \leq y} \mu(q,y) \right)  v_X[a_0=\hat{0}] \prod_{i=1}^{|X|} [a_i=e_i]\\
												 &\{ X \text{ has a join by assumption since it is in the support} \} \\
												 &= \sum_{X \in \supp(\vec{v})} \left( \sum_{e_1 \vee \ldots \vee e_{|X|} \leq q \leq y} \mu(q,y) \right) v_X\\
												 &\{ \text{M\"obius inversion} \} \\
												 &= \sum_{X \in \supp(\vec{v})} [e_1\vee\ldots\vee e_{|X|} = y]  v_X\\
												 &\{ \text{By the assumption stated in~\eqref{eqn:reqprop}} \} \\
												 &= v_X
\end{align*}
\end{proof}

\begin{lem}
\label{lem:solhom}
The mapping $h$ is a homomorphism from $\fieldF[(2^U,\cup)]$ to $\fieldF[P]$.
\end{lem}

\begin{proof}
For every $\vec{v},\vec{w} \in \bbz[(2^U,\cup)]$ we have
\begin{align*}
	h(\vec{v} + \vec{w}) &= \bigsolplus_{X \subseteq U} \st{v_X+w_X}{\hat 0} \soltimes \bigsoltimes_{e \in X} \st{1}{e}\\
											 &= \bigsolplus_{X \subseteq U} (\st{v_X}{\hat 0} \solplus \st{w_X}{\hat 0}) \soltimes \bigsoltimes_{e \in X} \st{1}{e})\\
											 &= \bigsolplus_{X \subseteq U} \biggl(\st{v_X}{\hat 0} \soltimes \bigsoltimes_{e \in X} \st{1}{e} \solplus \st{w_X}{\hat 0} \soltimes \bigsoltimes_{e \in X} \st{1}{e}\biggr)\\
											 &= \left(\bigsolplus_{X \subseteq U} \st{v_X}{\hat 0} \soltimes \bigsoltimes_{e \in X} \st{1}{e}\right) \solplus \left(  \bigsolplus_{X \subseteq U} \st{w_X}{\hat 0} \soltimes \bigsoltimes_{e \in X}\st{1}{e}\right)\\
											 &= h(\vec{v})\solplus h(\vec{w}).\\
\end{align*}	
\begin{align*}
	h(\vec{v} * \vec{w}) &= \bigsolplus_{X \subseteq U} \st{(\vec{v}*\vec{w})_X}{\hat 0} \soltimes \bigsoltimes_{e \in X} \st{1}{e}\\
 											 &\{\text{Definition of multiplication in } \fieldF[(2^U,\cup)] \}\\
											 &= \bigsolplus_{X \subseteq U} \biggl\langle\sum_{V \cup W = X}v_Vw_W,\hat 0\biggr\rangle\soltimes \bigsoltimes_{e \in X} \st{1}{e}\\
											 &\{ \text{By definition of $\solplus$}  \}\\
											 &= \bigsolplus_{X \subseteq U} \bigsolplus_{V \cup W = X}\st{v_Vw_W}{\hat 0}\soltimes \bigsoltimes_{e \in X} \st{1}{e}\\
											 &\{ \text{By Lemma \ref{lem:idempot} and commutativity of $\soltimes$ from Lemma \ref{lem:solring}}  \} \\
											 &= \bigsolplus_{X \subseteq U} \bigsolplus_{V \cup W = X}\st{v_Vw_W}{\hat 0}\soltimes \bigsoltimes_{e \in V} \st{1}{e} \soltimes \bigsoltimes_{e \in W} \st{1}{e}\\
											 &\{\text{Using that }V \text{ and } W \text{ determine }X \}\\
											 &= \bigsolplus_{V,W \subseteq U} \st{v_Vw_W}{\hat 0}\soltimes \bigsoltimes_{e \in V} \st{1}{e} \soltimes \bigsoltimes_{e \in W} \st{1}{e}\\
											 &\{ \text{ By Lemma \ref{lem:idempot} }  \}\\
											 &= \bigsolplus_{V,W \subseteq U} (\st{v_V}{\hat 0}\soltimes\st{w_W}{\hat 0})\soltimes \bigsoltimes_{e \in V} \st{1}{e} \soltimes \bigsoltimes_{e \in W} \st{1}{e}\\
											 &\{ \text{ By distributivity from Lemma \ref{lem:solring}}  \}\\
											 &= \left(\bigsolplus_{X \subseteq U} \st{v_X}{\hat 0} \soltimes \bigsoltimes_{e \in X} \st{1}{e}\right) \soltimes \left(\bigsolplus_{X \subseteq U} \st{v_X}{\hat 0}\soltimes \bigsoltimes_{e \in X} \st{1}{e}\right)\\
											 &= h(\vec{v})\soltimes h(\vec{w}).
\end{align*}
\end{proof}

\subsubsection{The algorithm for Lemma~\ref{thm:unionhashing}}
\label{sec:towunionhash}

We are now ready to give the proof of Lemma~\ref{thm:unionhashing}. As mentioned before, the proof idea is to construct a poset $P$ from the given set family and hash the given circuit to a circuit $C'$ over the Solomon algebra $\fieldF[P]$ with the homomorphic hash function $h$. Then by Theorem~\ref{thm:zethom} and Lemma~\ref{lem:solhom}, the zeta-transform of the output of $C'$ can be computed fast using point-wise multiplication, and then using M\"obius inversion the original output can be computed. Because of the hash property from Lemma~\ref{lem:homhashunion} of $h$, the required output of $C$ can then be read from the output of $C'$.

\begin{proof}[of Lemma~\ref{thm:unionhashing}]
We start by giving the algorithm. We assume without loss of generality that $\emptyset \in \mathcal{F}$.

\begin{algorithm}[H]
\begin{algorithmic}[1]
\REQUIRE $\mathtt{find}(C,\mathcal{F})$
	\STATE Construct the poset $P=(\mathcal{F}, \subseteq)$.
	\STATE Construct the circuit $C_1$ over $\fieldF[P]$ obtained from $C$ by applying $h$.
	\FOR{every $x \in P$}
		\STATE $w_x \leftarrow \mathtt{sub}(C_1,x)$.
	\ENDFOR
	\RETURN $\mathtt{mobius}(P,\vec{w})$
\vspace{0.5em}		
\REQUIRE $\mathtt{sub}(C_1,x)$
		\STATE Construct the circuit $C_2$ over $\fieldF$ obtained from $C_1$ by applying $\zeta_x: \vec{e} \mapsto \sum_{y \leq x} e_y$.
		\STATE Evaluate $C_2$ and return the output.
\vspace{0.5em}		
\REQUIRE $\mathtt{mobius}((P,\leq),\vec{w})$
		\STATE Let $P=\{v_1,v_2\ldots,v_{|P|}\}$ such that $v_i \leq v_i$ implies $i \geq j$.
		\STATE $\vec{z} \leftarrow \vec{w}$.
		\FOR{$i= 1,2,\ldots,|P|$}
			\FOR{every $v_j \leq v_i$}
				\STATE $z_i = z_i - z_j$
			\ENDFOR
		\ENDFOR
		\RETURN $\vec{z}$.
\end{algorithmic}
\end{algorithm}

The algorithm is Algorithm $\mathtt{find}$. Let us first analyse the complexity. Recall that to establish the claimed bounds we assume that each operation in $\fieldF$ takes one time unit. Step 1 can easily be implemented in $\OHT(|\mathcal{F}|)$ time. Step 2 expands every singleton (of $\fieldF[(2^U,\cup)]$) in $C$ according to \eqref{eq:hdef} into a product of at most $n+1$ singletons of $\fieldF[P]$, and thus can be implemented in time and storage $\OHT(n|C|)$. Step 6 can be implemented in time and space $\OHT(|C_1|)$ because the singleton $\vec{e}=\st{v}{y}$ is mapped to $[y \leq x]v$. In Step 7, the evaluation of $C_2$ over $\fieldF$ uses $|C_2|$ operations in $\fieldF$. The complexity of algorithm $\mathtt{mobius}$ is easily seen to be $\mathcal{O}^*(|\mathcal{F}|^2)$, and since the loop at Step 3 is the only other time-consuming part the resource bounds are clearly met.

We proceed with the proof of correctness of Algorithm $\mathtt{find}$. First note that since $\supp(\vec{v}) \subseteq \mathcal{F}$, we have by construction that for every $X=\{e_1,\ldots,e_{|X|}\} \in \supp(\vec{v})$, $\vee \{\{e_1\},\ldots,\{e_{|X|}\}\} = X$, since every other upper bound in $P$ on $\{\{e_1\},\ldots,\{e_{|X|}\}\}$ must be above the element $X$ of $P$. This directly implies that the requirement imposed by~\eqref{eqn:reqprop} is met by $P$. Hence by Lemma~\ref{lem:homhashunion} we have that $h(\vec{v})_{\vee \{\{e\} : e \in X\}}=v_X$ for every $X \in \supp(\vec{v})$. Thus to prove correctness it suffices to show that the vector $\vec{w} \zt$ is equal to $h(\vec{v})$.

Since $h$ is a homomorphism by Lemma~\ref{lem:solhom}, Observation~\ref{obs:homcir} implies that $C^1$ outputs $h(\vec{v})$. Then, we claim that for every $x \in P$ it holds that $w_x=\mathtt{sub}(C_1,x)=(h(\vec{v}) \zt)_x$. To see this note that on Line~7 first the zeta transform is applied to $C_1$ which is a homomorphism to $\fieldF^{P}$ by Theorem~\ref{thm:zethom}, and then we restrict to the coordinate $x$. The claim than follows by Observation~\ref{obs:homcir} since both transformations are homomorphisms.

Finally, since $w_x = (h(\vec{v}) \zt)_x$, we know that $(\vec{w} \mt)=h(\vec{v})$. The lemma then follows since it is easy to see that $\mathtt{mobius}(P,\vec{w})=\vec{w} \mt$.
\end{proof}


\subsubsection{Proof of Theorem~\ref{thm:unionhashsmallsupp}(b)}
The proof of Item (b) is fairly similar to the proof of Item (a). The main difference is the use of the poset $P$. In this section, we will use a poset larger than the support of $\vec{v}$, but the advantages are that 
\begin{enumerate}
	\item we not need to construct the poset using Iterative compression, and hence we do not need to determine all components of $\vec{v}$,
	\item the poset will be decomposable, implying that the bottleneck in (a), the M{\"o}bius inversion step, can be improved.
\end{enumerate}

Let us now start with the more formal treatment. Given two posets $(P,\leq_P)$ and $(Q,\leq_Q)$, the \emph{direct product} $P \times Q$ is the poset on the set $\{(p,q): p \in P \wedge q \in Q\}$ such that $(p_1,q_1) \leq (p_2,q_2)$ if $p_1 \leq_P p_2$ and $q_1 \leq_Q q_2$.

\begin{lem}[see \cite{stanleyenumerative}, Proposition 3.8.2]\label{lem:dirprodmob}
Let $P$ and $Q$ be posets and let $\mu_P$ and $\mu_Q$ denote their respective M{\"o}bius functions. Denote $\mu_{P \times Q}$ and $\leq$ for the M{\"o}bius function and order of $P \times Q$, then for $(p_1,q_1) \leq (p_2,q_2)$ it holds that
\[
	\mu_{P \times Q}((p_1,q_1),(p_2,q_2)) = \mu_P(p_1,p_2) \mu_Q(q_1,q_2).
\]
\end{lem}

\begin{lem}[\cite{yates1937design},``Yates' algorithm'']\label{lem:yatalg}
Given a vector $\vec{v} \in \fieldF^{2^U}$, there is an algorithm that computes the vector $\vec{v}\mt$ in $\OH(2^{|U|}|U|)$ time.
\end{lem}

We will also need the following small modification of Yates' algorithm:

\begin{lem}\label{lem:yatmob}
There is an algorithm that, given a poset $P \subseteq 2^{U}$ ordered by set inclusion and containing $\emptyset$ and $U$, computes a table with $\mu(0,S)$ and $\mu(S,U)$ for every  $S \in P$ in $\OH(|P|\cdot |U|)$ time and space
\end{lem}
\begin{proof}
In this proof $[i]$ denotes the set $\{1,\ldots,i\}$ for an integer $i$. Let $U=[n]$. We first define $g: [n] \rightarrow  \mathbb{N}$ by letting
\[
	g(S) =\begin{cases}
	1	&\text{if } S = \emptyset,\\
	\displaystyle \sum_{\substack{ Y \in P \\ Y \subset S}}- g(Y) & \text{ otherwise}.
	\end{cases}
\]
It is easy to see that $g(S)=\mu(\emptyset,S)$.  Also, define for every $i=1,\ldots,n$:
\[
f_i(S) = \begin{cases}
1 & \text{if } i=0 \text{ and } S=\emptyset\\
-[S \in P]g(S) &\text{if } i=0 \text{ and } S\neq\emptyset\\
f_{i-1}(S) + [i \in S]f_{i-1}(S \setminus i) &\text{otherwise}.
\end{cases}
\]
Then it holds that $g(S)=f_n(S)$ for every $S \subseteq U$. Note that by symmetry (that is, reversing the order of the poset), this procedure can be used as well to compute $\mu(S,U)$ for every  $S \in P$ in the same time bounds.
\end{proof}

\begin{lem}\label{lem:trd}
There is an algorithm that, given an integer $s$ and a circuit $C$ with singleton inputs in $\fieldF[(2^U,\cup)]$ outputting $\vec{v}$, computes $v_{U}$ in time 
\[
	\OHS(\max\{2^s,|\supp(\phi_s)|2^{n-s}n^{\OH(1)}\}).
\]
\end{lem}
\begin{proof}
Assume $U=\{e_1,e_2,\ldots,e_n\}$, and let $C^{s}$ be the restriction of $C$ to $U_s=\{e_1,\ldots,e_s\}$ and $\vec{v^s}$ be the output of $C^s$. 

\begin{clm}\label{clm:incexc}
There is an algorithm computing $\supp(\vec{v^s})$ in $\OHS(2^sn^{\OH(1)})$ time.
\end{clm}
\begin{proof}
By Theorem~\ref{thm:zethom}, $\zt$ is a homomorphism from $\fieldF[P]$ to $\fieldF^P$. Hence, similarly as in the proof of Theorem~\ref{thm:unionhashsmallsupp}(a), we can compute $\vec{v^s}\zt$ in time polynomial in the input size. Then the claim follows directly by applying Lemma~\ref{lem:yatalg}.
\end{proof}

Now the algorithm is as follows

\begin{algorithm}[H]
\begin{algorithmic}[1]
\REQUIRE $\mathtt{find2}(C)$
	\STATE Obtain $\supp(\vec{v^s})$ using Claim~\ref{clm:incexc}
	\STATE Let $P' = (\supp(\vec{v^s}),\subseteq)$ and $Q=(2^{U \setminus U_s},\subseteq)$ and construct the poset $P = P'\times Q$.
	\STATE Construct the circuit $C_1$ over $\fieldF[P]$ obtained from $C$ by applying the homomorphism $h$ as defined in~\eqref{eq:hdef}.
	\FOR{every $x \in P$}
		\STATE $w_x \leftarrow \mathtt{sub}(C_1,x)$.
	\ENDFOR
	\STATE Compute $\mu(0,x)$ for every $x \in P'$ using Lemma~\ref{lem:yatmob}.
	\RETURN $\sum_{x=(X_1,X_2) \in P}w_x\mu(X_1,U_s)(-1)^{|(U \setminus U_2) \setminus X_2|}$
\end{algorithmic}
\end{algorithm}

The arguments for the correctness of this algorithm are all similar to the arguments of for the correctness of Algorithm $\mathtt{find2}$: first note that $\supp(v) \subseteq \{X_1 \cup X_2: X_1 \in P' \wedge X_2 \in Q\}$. This in turn implies the condition of Lemma~\ref{lem:homhashunion}, and hence if $\vec{v'}$ is the output of $C^1$ we indeed have that $v'_U=v_U$. As shown before in the proof of Lemma~\ref{thm:unionhashing}, $\mathtt{sub}(C_1,x)=(\vec{v'}\zt)_x$. Hence it remains the show that on Line~7 the expression is indeed the M{\"o}bius inversion formula. This follows from the definition of $\mt$, the direct product property of the M{\"o}bius function from Lemma~\ref{lem:dirprodmob} and the fact that in the subset lattice $\mu(X,Y)=(-1)^{|Y \setminus X}|$ for $X \subseteq Y$.

For the running time of the above algorithm, note that Line 1 takes $\OHS(2^sn^{\OH(1)})$ time. For Line 2, we have $|P'| \leq |\supp(\vec{v})|$ and $|Q| = 2^{n-s}$, and hence $|P| \leq |\supp(\vec{v})|2^{n-s}$. Thus, Lines 2-5 take $\OHS(|\supp(\vec{v})|2^{n-s}n^{\OH(1)})$ time. Line 6 takes $\OHS(2^sn^{\OH(1)})$ time, and using these values, Line 7 can be performed in time at most $\OHS(|\supp(\vec{v})|2^{n-s}n^{\OH(1)})$ .
\end{proof}
Given Lemma~\ref{lem:yatmob}, the only thing left to prove Theorem~\ref{thm:unionhashsmallsupp}(b) is to solve the technical issue that $\alpha$ is not given:

\begin{proof}[of Theorem~\ref{thm:unionhashsmallsupp}(b)] Simultaneously try all integers $0 < s \leq n/2$, and for every $s$ run the algorithm of Lemma~\ref{lem:trd}. Terminate whenever any of these algorithms terminates. This procedure runs in time 
\[
	\OHS\big(\min_{0 < s \leq n/2} \max\{2^s,|\supp(\phi_s)|2^{n-s}n^{\OH(1)}\}\big) \leq \OHS\big(2^{(1-\alpha/2)n}\big),
\]
where the inequality is achieved by taking $s=\left\lfloor (1-\alpha/2)n \right\rfloor$.
\end{proof}

\end{document}